\documentclass[conference]{IEEEtran}

\usepackage[utf8]{inputenc} 
\usepackage[T1]{fontenc}
\usepackage{amsmath,amssymb}

\usepackage{url}
\usepackage{tcolorbox}
\usepackage{cite}
\usepackage[hang,flushmargin]{footmisc}

\usepackage{soul,times,amsfonts,amssymb,amsthm,mathrsfs,xcolor,graphicx,enumitem,booktabs,euscript,mathtools,nicefrac,float,graphicx,enumitem}
\usepackage{cite}
\usepackage[font=footnotesize]{subcaption}
\usepackage[draft,unicode,bookmarks=false]{hyperref}
\hypersetup{
    colorlinks,
    linkcolor={blue!80!black},
    citecolor={blue!80!black},
    urlcolor={blue!80!black}
}
\usepackage[capitalize]{cleveref}
\newcommand\redout{\bgroup\markoverwith{\textcolor{red}{\rule[0.5ex]{2pt}{0.8pt}}}\ULon}

\newtheorem{theorem}{Theorem}[section]
\newtheorem{lemma}[theorem]{Lemma}

\newtheorem{corollary}[theorem]{Corollary}

\theoremstyle{remark}
\newtheorem{remark}{Remark}
\newtheorem{definition}{Definition}[section]

\newcommand\nc\newcommand
\nc\bfa{{\boldsymbol a}}\nc\bfA{{\boldsymbol A}}\nc\sA{{\EuScript A}}\nc\cA{{\mathcal A}}
\nc\bfb{{\boldsymbol b}}\nc\bfB{{\boldsymbol B}}\nc\cB{{\mathcal B}}\nc\sB{{\EuScript B}}
\nc\bfc{{\boldsymbol c}}\nc\bfC{{\boldsymbol C}}\nc\cC{{\mathscr C}}
\nc\bfd{{\boldsymbol d}}\nc\bfD{{\boldsymbol D}}\nc\cD{{\mathscr D}}
\nc\bfe{{\boldsymbol e}}\nc\bfE{{\boldsymbol E}}\nc\cE{{\mathcal E}}\nc\sE{{\mathscr E}}
\nc\bff{{\boldsymbol f}}\nc\bfF{{\boldsymbol F}}\nc\cF{{\mathscr F}}
\nc\bfg{{\boldsymbol g}}\nc\bfG{{\boldsymbol G}}\nc\cG{{\EuScript G}}\nc\sG{{\mathscr G}}
\nc\bfh{{\boldsymbol h}}\nc\bfH{{\boldsymbol H}}\nc\cH{{\mathcal H}}\nc\sH{{\mathscr H}}
\nc\bfi{{\boldsymbol i}}\nc\bfI{{\boldsymbol I}}\nc\cI{{\EuScript I}}\nc\sI{{\mathscr I}}
\nc\bfj{{\boldsymbol j}}\nc\bfJ{{\boldsymbol J}}\nc\cJ{{\EuScript J}}
\nc\bfk{{\boldsymbol k}}\nc\bfK{{\boldsymbol K}}\nc\cK{{\EuScript K}}
\nc\bfl{{\boldsymbol l}}\nc\bfL{{\boldsymbol L}}\nc\cL{{\EuScript L}}
\nc\bfm{{\boldsymbol m}}\nc\bfM{{\boldsymbol M}}\nc\cM{{\EuScript M}}\nc\sM{{\mathscr M}}
\nc\bfn{{\boldsymbol n}}\nc\bfN{{\boldsymbol N}}\nc\cN{{\EuScript N}}
\nc\bfo{{\boldsymbol o}}\nc\bfO{{\boldsymbol O}}\nc\cO{{\EuScript O}}\nc\sO{{\mathscr O}}
\nc\bfp{{\boldsymbol p}}\nc\bfP{{\boldsymbol P}}\nc\cP{{\EuScript P}}\nc\sP{{\mathscr P}}
\nc\bfq{{\boldsymbol q}}\nc\bfQ{{\boldsymbol Q}}\nc\cQ{{\mathcal Q}}
\nc\bfr{{\boldsymbol r}}\nc\bfR{{\boldsymbol R}}\nc\cR{{\EuScript R}}\nc\sR{{\mathscr R}}
\nc\bfs{{\boldsymbol s}}\nc\bfS{{\boldsymbol S}}\nc\cS{{\EuScript S}}
\nc\bft{{\boldsymbol t}}\nc\bfT{{\boldsymbol T}}\nc\cT{{\EuScript T}}\nc\sT{{\mathscr T}}
\nc\bfu{{\boldsymbol u}}\nc\bfU{{\boldsymbol U}}\nc\cU{{\EuScript U}}
\nc\bfv{{\boldsymbol v}}\nc\bfV{{\boldsymbol V}}\nc\cV{{\mathscr V}}
\nc\bfw{{\boldsymbol w}}\nc\bfW{{\boldsymbol W}}\nc\cW{{\mathscr W}}\nc\sW{{\mathscr W}}
\nc\bfx{{\boldsymbol x}}\nc\bfX{{\boldsymbol X}}\nc\cX{{\EuScript X}}
\nc\bfy{{\boldsymbol y}}\nc\bfY{{\boldsymbol Y}}\nc\cY{{\mathscr Y}}
\nc\bfz{{\boldsymbol z}}\nc\bfZ{{\boldsymbol Z}}\nc\cZ{{\EuScript Z}}
\usepackage{dsfont}
\newcommand{\1}{\mathds{1}}
\nc{\remove}[1]{}

\allowdisplaybreaks

\DeclareSymbolFont{bbold}{U}{bbold}{m}{n}
\DeclareSymbolFontAlphabet{\mathbbold}{bbold}

\DeclareMathOperator{\qsc}{{\text{$q$}}SC}
\DeclareMathOperator{\qec}{{\text{$q$}}EC}

\newcommand{\R}{{\mathbb R}}

\newcommand{\Z}{{\mathbb Z}}

\nc\torus{{\mathbb{T}}}
\nc\reals{{\mathbb R}}
\nc{\ff}{{\mathbb F}}
\nc{\PP}{{\mathbb P}}
\nc{\complex}{{\mathbb C}}

\nc\inftyeq{\overset{\infty}{=}}

\usepackage{tikz}
\usetikzlibrary{positioning}
\usetikzlibrary{calc}
\usetikzlibrary{decorations.pathreplacing}

\definecolor{newred}{HTML}{ff382e}
\definecolor{newgreen}{HTML}{549641}
\definecolor{newblue}{HTML}{4c4cfc}
\definecolor{neworange}{HTML}{c98702}

\nc\Renyi{R{\'e}nyi }
\usepackage[capitalize]{cleveref}
\begin{document}

\title{Universally optimal (wiretap) codes}

\author{%
   \IEEEauthorblockN{{\sc Madhura Pathegama}  \hspace*{.8in}   {\sc Alexander Barg} \\
   \hspace*{.5in} Georgia Tech., Atlanta, GA, USA \hspace*{.3in} U. of Maryland, College Park, USA    \\ \hspace*{.3in} \{pankajap,abarg\}@umd.edu 
}
 }

\maketitle



 \setlength{\abovedisplayskip}{1pt}
 \setlength{\belowdisplayskip}{1pt}

\begin{abstract} 
Universally optimal (UO) codes were introduced by H. Cohn and A. Kumar in 2007 and later extended to the discrete setting by Cohn and Y. Zhao. They minimize the ``energy'' among all codes of the same size for a certain class of potential functions. So far only a small number of specific functionals have been linked to information theory problems. We add one more, showing that UO codes optimize $\alpha$-mutual information for $\alpha=2$ in the context of wiretap channels with a noiseless main channel. This implies that UO codes are optimal for transmission over this type of wiretap channels.
\end{abstract}

\section{Introduction}
Wyner’s wiretap channel model involves three terminals, a transmitter $A$, a legitimate receiver $B$, and an eavesdropper $E$ \cite{wyner1975wire}. The goal of $A$ is to communicate a uniformly distributed message
$M$ drawn from a finite set $\sM$ to $B$ over a channel
$\sW_B$, while $E$ observes this transmission on the output of another channel, $\sW_E$.

We assume that $M \in \sM$ is encoded into a subset of ${\Z}_q^n$, and a random
element, $X$, from this subset is transmitted over $n$ uses of the channel $\sW_B$.
The legitimate receiver observes an output sequence $Y$, while the eavesdropper
observes a sequence $Z$ at the output of $\sW_E$.
The encoder is designed to ensure reliable communication to $B$ while
simultaneously limiting the information revealed to $E$.

A standard and widely used approach for achieving both reliability and secrecy
is the use of coset coding schemes \cite{wyner1975wire}. Such schemes rely on a pair of additive codes
$C_E \subseteq  C_B \subseteq  {\Z}_q^n$, where the message set $\sM$ is identified with the
quotient $C_B / C_E$. Thus, each message corresponds to a coset of $C_E$ in
$C_B$, and the transmitted sequence is chosen uniformly at random from this coset.
This randomization is crucial for reducing the statistical dependence between
the transmitted message and eavesdropper’s observation.

In this work, we focus on a special but important case in which the main channel
$\sW_B$ is noiseless \cite{subramanian2011strong}, \cite{mahdavifar2011achieving}; see \cref{fig:model}.
In this setting, reliable communication is immediate, and the code $C_B$
may be taken to be the entire space ${\Z}_q^n$.
The primary design objective then becomes secrecy: information about $M$ acquired by $E$ by observing $Z$
must be minimized. While most published works on the wiretap channel considered the asymptotic version of the transmission problem for $n\to\infty$, here we focus on the finite-length case and strict optimality.


\begin{figure}[t]
\centering
\begin{tikzpicture}[>=latex,scale=1.0]

\node[draw, minimum width=1.8cm, minimum height=1cm] (enc) at (-.3,0) {Encoder};
\node (msg) at (-2.2,0) {$M$};
\node[draw, minimum width=1.7cm, minimum height=1cm] (main) at (3,1.0) {$\sW_B$};
\node[draw, minimum width=1.7cm, minimum height=1cm] (evech) at (3,-1.0) {$\sW_E$};
\node (rx) at (5.3,1.0) {$Y=X$};
\node (eve) at (5.5,-1.0) {$Z$};

\draw[->] (msg) -- (enc);
\draw[->] (enc) -- node[above] {$X$} (main);
\draw[->] (enc) -- node[below] {$X$} (evech);

\draw[->] (main) -- (rx);
\draw[->] (evech) -- (eve);

\node at (3.2,-2.0) {$\qsc(p)$ or $\qec(\epsilon)$};

\end{tikzpicture}
\caption{The wiretap channel}
\label{fig:model}
\end{figure}
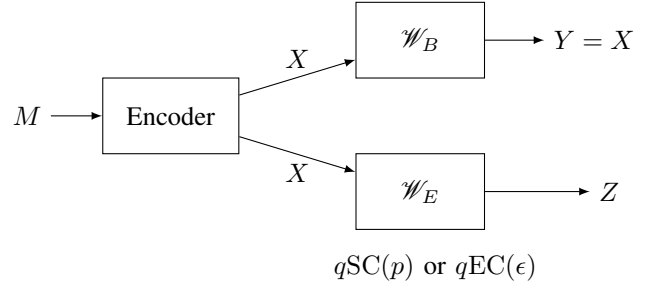

An interesting question raised in \cite{swain2025best} is whether there exist coset
coding schemes that simultaneously minimize the mutual information $I(M;Z)$ among all codes of the same size and for all
memoryless erasure channels. It was shown that certain classical code families—such as Hamming codes, Hadamard codes, and first-order Reed–Muller codes—satisfy this universality property. In this paper, we address a related but slightly different question.
Rather than focusing on Shannon mutual information, we seek codes that universally minimize a range of dependence measures based on Rényi entropy and Rényi divergence across broad classes of eavesdropper channels. Rényi-based leakage measures, including Rényi, Sibson, and Arimoto mutual information, play a central role in modern analyses of information-theoretic security. These measures provide stronger notions of secrecy than Shannon mutual information  \cite{iwamoto2014information,yu2018renyi,pathegama2023smoothing} and are closely tied to operational tasks such as hypothesis testing \cite{shayevitz2011renyi,sason2017arimoto} and worst-case distinguishability \cite{alimomeni2012guessing,issa2016operational}.

In this work, we study wiretap codes that minimize the Rényi mutual information,
Sibson mutual information, and Arimoto mutual information of order $\alpha=2$
for all memoryless $q$-ary symmetric channels ($\qsc$'s) and $q$-ary erasure channels at the
eavesdropper.

{\em Universally optimal {\rm (UO)} codes} were originally introduced in the study of energy minimization problems 
\cite{cohn2007universally,cohn2014energy}, \cite{boyvalenkov2016universal}. Their links to information theory were limited to minimizing
the probability of undetected error over a $\qsc$ \cite{ashikhmin1999binomial}. 
Here we identify another information-theoretic motivation for this code family, showing that they simultaneously minimize all the dependence measures mentioned above for $\alpha=2$. For the $\qsc$, this universality also extends to the case
$\alpha = \infty$.

\vspace*{-.05in}
\section{Preliminaries}

\subsection{Codes and channels}
\vspace*{-.05in}
A code $C\subseteq \Z_q^n$ is called {\em additive} if it forms an additive subgroup of $(\Z_q^n)^+$. 
We consider transmission over the wiretap channel as described above using the coset coding scheme that relies on additive codes. Two types of memoryless wiretapper's channels $\sW_E$ with input alphabet $\Z_q$ are considered: the $\qsc(p)$ channel $\Z_q\to\Z_q$ defined as
    \begin{align*}
     \qsc(p): P(b|a)&=(1-p)\delta_{b,a}+\frac p{q-1}(1-\delta_{b,a})
     \end{align*}
and channels with erasures $\Z_q\to (\Z_q\cup ?)$, defined in Sec.~\ref{sec: qEC}.

Since the action of $\sW_E$ does not depend on the input, the probability of sending $x\in \Z_q^n$ and receiving $z\in \Z_q^n$ for $\qsc(p)$ depends only on the Hamming distance $d(z,x)$. Letting $u=z-x$, we can write $P(u)$ as a {\em radial function}
   $$
   \beta_p^{q,n}(u)=\beta_p^{q,n}(|u|)=(1-p)^{n-|u|}\Big(\frac p{q-1}\Big)^{|u|},
   $$
where $|u|$ denotes the Hamming weight. We may omit the mention of $q$ and $n$ if they are clear from the context. The uniform distribution over the code is denoted $P_C$ below.

\subsection{UO Codes and Completely Monotonic Potentials}

Let $C \subseteq {\Z}_q^n$ be a code and let $d(x,y)$ be the
Hamming distance.
Given a (radial) potential function $f : \{0,1,\dots,n\} \to \mathbb{R}$,
the \emph{total potential energy} of $C$ is defined as \cite{yudin1992minimum}
\begin{align}
    V_f(C)
    = \frac{1}{|C|}\sum_{x,y \in C} f\!\big( d(x,y) \big).
\end{align}

A function $f:\{0,\dots,n\} \to \mathbb{R}$ is said to be
\emph{completely monotonic} if its finite differences alternate in sign:
\begin{align}
    (-1)^k \Delta^k f(i) \;\ge\; 0,
    \qquad \forall\, k\ge 0, 0\le i\le n-k
\end{align}
where $\Delta f(i) = f(i+1) - f(i)$ and $\Delta^k f $ denotes the $k$-th finite difference. 

\begin{definition}[\textbf{Universal optimality}]
A code $C \subseteq {\Z}_q^n$ is said to be
\emph{universally optimal} if for every completely monotonic potential $f$
and every code $C'\subseteq \Z_q^n$ with $|C'|=|C|$,
\begin{align}
    V_f(C) \;\le\; V_f(C').
\end{align}
\end{definition}
Thus, a UO code simultaneously minimizes the energy for \emph{all}
admissible potentials determined by the Hamming distance.
Equivalently, using the {\em distance distribution}
$A_i = \frac{1}{|C|}|\{(x,y)\in C^2: d(x,y)=i\}|$, we may write
\begin{align}
    V_f(C)
    = \sum_{i=0}^n A_i\, f(i),
\end{align}
so universal optimality means that $C$ minimizes this linear functional
for every completely monotonic $f$. Classical examples of universally optimal codes include Hamming, simplex, and MDS codes; a list of the known UO codes is provided in Table 1 of \cite{cohn2014energy}.

As shown in \cite{cohn2014energy}, the cone of completely monotonic
functions on $\{0,\dots,n\}$ is generated by
\begin{align}\label{eq: cm_basis}
    f_j(i) = \binom{n-i}{j}, \qquad j=0,\dots,n.
\end{align}
In other words, every completely monotonic function admits a nonnegative expansion
\begin{align*}
    f(i) = \sum_{j=0}^{n-i} c_j \binom{n-i}{j},
    \qquad c_j \ge 0.
\end{align*}
Thus, a code $C$ is universally optimal
if and only if it minimizes the energies corresponding to the potentials $f_j(i)$. The energies $V_{f_j}$ are called binomial moments of the code \cite{ashikhmin1999binomial}, where UO codes were called extremal (no connection to code energy was made there).

An easy calculation confirms that the potential $\beta_p$ is completely monotonic: for all $p \in (0,1-1/q)$,
\begin{align}
        \beta_p(i) = \left(\frac{p}{q-1}\right)^{n}\sum_{j=0}^{n-i} \left(\frac{q(1-p)-1}{p}\right)^{j} \binom{n-i}{j}.\label{eq: beta cm}
    \end{align}
Switching from the radial to a spatial view, note that the convolution $(\beta_p\ast\beta_p)(x), x\in \Z_q^n$  yields a potential
\begin{align}\label{eq:r(p)}
  \beta_{r(p)}, \quad    r(p)= 2p(1-p) + \frac{q-2}{q-1}\,p^2
\end{align}
which is completely monotonic by \eqref{eq: beta cm}. Note further that $r(p)$ represents a cascade of two $\qsc(p)$ channels, and thus is a monotone increasing function of $p\in[0,\frac{q-1}q)$.

\vspace*{-.1in}

\subsection{\Renyi measures of secrecy}

The Shannon mutual information $I(M;Z)$ is a widely used secrecy metric, as it
quantifies the statistical dependence between the message $M$ and the
eavesdropper’s observation $Z$. Going beyond $I(M;Z)$, 
several alternative notions of mutual information can be
defined using Rényi entropy–based quantities, which are often better suited
for capturing worst-case or higher-moment leakage.

We begin by recalling the definition of Rényi divergence.
For $\alpha \in (0,1) \cup (1,\infty)$, the \emph{Rényi divergence} of order $\alpha$
between two probability distributions $P$ and $Q$ on a common alphabet
$\cC$ is defined as
\[
    D_\alpha(P\|Q)
    := \frac{1}{\alpha-1}
       \log \sum_{x \in \cC}
       P(x)^{\alpha} Q(x)^{1-\alpha},
\]
whenever $\operatorname{supp}(P) \subseteq \operatorname{supp}(Q)$. (The base of the logarithm can be any number greater than 1, but for simplicity we set it to $q$.)
The limiting cases $\alpha \to 0$, $\alpha \to 1$, and $\alpha \to \infty$
are well defined and recover, respectively, the min-divergence,
the Kullback--Leibler divergence, and the max-divergence.

Using Rényi divergence, we define two commonly used $\alpha$-parametrized
generalizations of mutual information.
Let $U$ and $V$ be arbitrary random variables.
The \emph{Rényi mutual information} and  \emph{Sibson mutual information} 
of order $\alpha$ \cite{sibson1969information} are defined, respectively, as
\begin{align}
    I_\alpha^{\mathrm{R}}(U;V)
        &:= D_\alpha\!\left(P_{UV}\,\middle\|\,P_U P_V\right), \\
    I_\alpha^{\mathrm{S}}(U;V)
        &:= \min_{Q_V}
            D_\alpha\!\left(P_{UV}\,\middle\|\,P_U Q_V\right),
            \label{eq:Sibson}
\end{align}
where the minimization in \eqref{eq:Sibson} is over all probability
distributions $Q_V$ on the alphabet of $V$.

An explicit expression for the optimizer in \eqref{eq:Sibson} is known and
is given by \cite[Eq.(31)]{esposito2025sibson}
\begin{align}\label{eq:Sibson-opt}
    Q_V^\star(v)
    := \frac{
        \left(\sum_{u} P_U(u)\, P_{V|U}(v|u)^{\alpha}\right)^{1/\alpha}
    }{
        \sum_{v'}
        \left(\sum_{u} P_U(u)\, P_{V|U}(v'|u)^{\alpha}\right)^{1/\alpha}
    }.
\end{align}

Another widely used $\alpha$-parametrized dependence measure is the
\emph{Arimoto mutual information}\cite{arimoto1977information}, defined as
\begin{align}\label{eq:Arimoto}
    I_\alpha^{\mathrm{A}}(U;V)
        := H_\alpha(U) - H_\alpha^{\mathrm{A}}(U|V),
\end{align}
where
\begin{align*}
    H_\alpha(U)
        &:= \frac{1}{1-\alpha}
            \log \sum_{u} P_U(u)^{\alpha}, \\
    H_\alpha^{\mathrm{A}}(U|V)
        &:= \frac{\alpha}{1-\alpha}
            \log \sum_{v}
            \biggl(\sum_{u} P_{UV}(u,v)^{\alpha}\biggr)^{1/\alpha}.
\end{align*}
Here, $H_\alpha(\cdot)$ denotes the {Rényi entropy} of order $\alpha$, 
while $H_\alpha^{\mathrm{A}}(\cdot|\cdot)$ is referred to as the
{conditional Arimoto Rényi entropy} of order $\alpha$.


All three versions of $I_\alpha$ reduce to the Shannon mutual information in the limit of
$\alpha \to 1$, and each is a non-decreasing function of $\alpha$. By definition, the Rényi mutual information
$I_\alpha^{\mathrm{R}}(U;V)$ is symmetric in $(U,V)$, whereas the Sibson and
Arimoto versions generally are not.
Consequently, for a fixed $\alpha$, one may consider up to five distinct
dependence measures between the message $M$ and eavesdropper’s
observation $Z$, namely
\begin{equation}\label{eq: I-alpha's}
I_\alpha^{\mathrm{R}}(M;Z), \quad
I_\alpha^{\mathrm{S}}(U;V), \quad
I_\alpha^{\mathrm{A}}(U;V), 
\end{equation}
where $\{U,V\}=\{M,Z\}$ (any order).

See \cite{verdu2015alpha,aishwarya2019remarks,lapidoth2019two,esposito2025sibson}
for an expanded discussion of various notions of $\alpha$-mutual information.

In this work we study the cases $\alpha = 2,\infty$ for $\sW_E=\qsc$ and
$\alpha = 2$ for the erasure channel $\sW_E$.
Our analysis is primarily carried out in terms of Rényi mutual information.
However, in the channel settings considered here, the 
relationships among the different $\alpha$-mutual information measures
allow us to translate our results to the corresponding Sibson and Arimoto
secrecy metrics. 

\section{Eavesdropper's channel \texorpdfstring{$\qsc(p)$}{W}}

Consider a wiretap channel setting in which $\sW_B$ is noiseless and
$\sW_E$ is a $\qsc(p)$ channel.
The encoder employs a standard coset coding scheme based on an additive
code $C$.

In this model, the channel noise acts additively on the transmitted codeword
$X$, and the induced distribution of eavesdropper’s observation $Z$ can
be expressed as the convolution of the input distribution with the noise
distribution $\beta_p$.
Moreover, the marginal distribution $P_Z$ is uniform.
Owing to these structural and symmetry properties, we obtain the following result.

\begin{lemma}
For $\sW_E=\qsc(p)$, the five definitions of $I_\alpha$
in \cref{eq: I-alpha's} coincide for all $\alpha \in [0,\infty]$.
\end{lemma}
{\em Proof idea:} First, note that when the first random variable in the Arimoto mutual information is uniformly
distributed, it coincides with the Sibson mutual information \cite[Corollary~IV.16]{esposito2025sibson}. The equality between \Renyi and Sibson versions relies on the variational characterization of Sibson's information radius \eqref{eq:Sibson}-\eqref{eq:Sibson-opt}.

Below, we denote the common value of $\alpha$-information measures by $I_\alpha^\star$.
The main result of this section follows.

\begin{theorem}\label{thm: qsc-2}
Consider the wiretap channel described above, and suppose that $C$ is a UO code. Then the coset coding scheme based on $C$ minimizes $I_2^\star(M;Z)$ over all codes of the same size, simultaneously for all values of $p \in (0,1-1/q)$.
\end{theorem}

The proof proceeds by showing that the quantity appearing inside the
logarithm in the \Renyi mutual information admits an energy representation
with respect to the potential $\beta_{r(p)}=\beta_p^{\ast 2}$ defined in \cref{eq:r(p)}.
This representation is established in Lemma~\ref{lem: qsc-2} below.
Since $\beta_p$ is a completely monotonic potential, universal optimality of
the code $C$ implies the desired minimization property, thereby completing
the proof.

\begin{lemma}\label{lem: qsc-2}
Consider a coset encoding scheme based on an additive code $C$ with distance
distribution $(A_i)_{i=0}^n$.
Then
\begin{align}\label{eq: qsc-2}
    I_2^\star(M;Z)
    = \log\!\left( \frac{q^n}{|C|} \sum_{i=0}^n \beta_{r(p)}(i)\, A_i \right).
\end{align}
\end{lemma}

\begin{proof} Let $c_m$ be an arbitrary element of the coset corresponding to message $m$. We proceed as follows:
 \begin{align*}
 P_{Z|M}&(z \mid m)
 = \sum_{x \in {\Z}_q^n} P_{Z,X|M}(z,x \mid m)\\
& = \sum_{x \in {\Z}_q^n} P_{Z|X}(z \mid x)\, P_{X|M}(x \mid m) \nonumber\\
 &= \sum_{x \in {\Z}_q^n} \beta_p(z-x)\, P_C(x-c_m)
 = (\beta_p \ast P_C)(z-c_m).
 \end{align*}
   We obtain
\begin{align*}\label{eq: closed}
    I_\alpha^{\mathrm R}(M;Z) &= \frac{1}{\alpha-1}\log \sum_{m,z}\frac{P_M(m)P_{Z|M}(z|m)^\alpha}{P_Z(z)^{\alpha-1}} 
    \\
    &= \frac{1}{\alpha-1}\log \Big(q^{n(\alpha-1)}\sum_{z}(\beta_p \ast P_C)(z)^\alpha\Big).
\end{align*}

For $\alpha = 2$, we have:
\begin{align}
    I_2^\star(M;Z) = \log \Big(q^{n}\sum_{z}(\beta_p \ast P_C)(z)^2\Big).
\end{align}

We now expand the sum 
\begin{align*}
&\sum_{z} \big( \beta_p \ast P_C\big)(z)^2\\
&= \sum_{z} \big( \sum_{x} \beta_p(z-x)\, P_C(x) \big)
                  \big( \sum_{x'} \beta_p(z-x')\, P_C(x') \big) \\
&= \sum_{x,x'} P_C(x) P_C(x')
   \sum_{z} \beta_p(z-x)\, \beta_p(z-x').
\end{align*}

Since $\beta_p(z-x') = \beta_p(x'-z)$, the inner sum is a convolution,
and with \cref{eq:r(p)}, we obtain
\begin{align*}
\sum_{z} \left( \beta_p \ast P_C(z) \right)^2 &= \sum_{x,x'} P_C(x) P_C(x')
   \beta_{r(p)}(x-x')\\
   &= \sum_{x,x''} P_C(x) P_C(x-x'')
   \beta_{r(p)}(x'')\\
   &= \sum_{x''} \beta_{r(p)}(x'') \sum_{x} P_C(x) P_C(x''-x)\\
   &= \sum_{x''} \beta_{r(p)}(x'')P_C(x'')\\
   &= \frac{1}{|C|} \sum_{i=0}^n \beta_{r(p)}(i)\, A_i.\qedhere
\end{align*}
\end{proof}
\vspace*{-.1in}As noted earlier, the potential $\beta_p$ is completely monotonic.
Consequently, universally optimal codes minimize the expression in
\eqref{eq: qsc-2} for each $p$ among all codes of the same size.
This establishes Theorem~\ref{thm: qsc-2}.

An analogous result holds for the case $\alpha = \infty$.

\begin{theorem}\label{thm: qsc-inf}
Consider the wiretap channel described above, and suppose that $C$ is a UO code. Then the coset coding scheme based on $C$ minimizes $I_\infty^\star(M;Z)$ over all codes of the same size, simultaneously for all values of $p \in (0,1-1/q)$.
\end{theorem}

In analogy with Lemma~\ref{lem: qsc-2}, we show that $I_\infty^\ast$ can be written as an energy function with 
the potential $\beta_p$.
This representation is established in the following lemma, which directly
implies Theorem~\ref{thm: qsc-inf}.

The proof uses the Fourier transform on the group ${\Z}_q^n$.
For a function $F : {\Z}_q^n \to {\R}$, its Fourier transform is defined as
\begin{align*}
    \widehat{F}(u)
    = \frac{1}{q^n} \sum_{x \in {\Z}_q^n} F(x)\omega_q^{-u\cdot x},
\end{align*}
with the inversion formula
\begin{align*}
    F(x)
    = \sum_{u \in {\Z}_q^n} \widehat{F}(u)\omega_q^{u\cdot x}.
\end{align*}

Here $\omega_q$ is a $q$th root of unity and $u\cdot x := \sum_iu_ix_i \mod q$.
As usual, the Fourier transform exchanges convolution with a product:
\begin{align}\label{eq: conv}
    \widehat{F \ast G}(u)
    = q^n\, \widehat{F}(u)\widehat{G}(u).
\end{align}
Moreover, the functions we use below are even ($F(x)=F(-x)$), and thus their Fourier transforms are real.

\begin{lemma}\label{lem: qsc-inf}
    Consider a coset encoding scheme based on an additive code $C$ with distance
distribution $(A_i)_{i=0}^n$.
Then
    \begin{align}\label{eq: qsc-inf}
        I_\infty^\star(M;Z)
        = \log\!\left( \frac{q^n}{|C|}\sum_{i=0}^n \beta_{p}(i)\, A_i \right).
    \end{align}
\end{lemma}
\begin{proof} 
    Letting $\alpha \to \infty$ in \eqref{eq: closed}, we obtain
    \begin{align}\label{eq: qsc-inf2}
        I_{\infty}^{\mathrm{R}}(M;Z) = \log (q^n \max_z (\beta_p \ast P_C)(z)).
    \end{align}
To find the maximum in \cref{eq: qsc-inf2}, observe that
    \begin{align*}
        (\beta_p \ast P_C)(0)
        &= q^n \sum_{u \in {\Z}_q^n}
            \widehat{\beta_p}(u)\, \widehat{P_C}(u).
    \end{align*}
Replacing $p$ with $r^{-1}(p)$ in the definition of $\beta_{r(p)}$ \eqref{eq:r(p)}, we obtain $\beta_p = \beta_{r^{-1}(p)} \ast \beta_{\,r^{-1}(p)}$. Now \cref{eq: conv} implies that $\widehat{\beta_p}$ is nonnegative. Likewise, since $P_C = P_C \ast P_C$, $\widehat{P_C}$ is also nonnegative. Consequently, all products of the form $\widehat{\beta_p}(u)\, \widehat{P_C}(u)$ are nonnegative.

    For a general $z$,
    \begin{align*}
        (\beta_p \ast P_C)(z)
        &= \Big|q^n \sum_u \widehat{\beta_p}(u)\widehat{P_C}(u)\omega_q^{u\cdot z}\Big|\\
        &\leq q^n \sum_u \big|\widehat{\beta_p}(u)\widehat{P_C}(u)\omega_q^{u\cdot z}\big|\\
        &= q^n \sum_u \widehat{\beta_p}(u)\widehat{P_C}(u)
        = (\beta_p \ast P_C)(0),
    \end{align*}
    so the convolution attains its maximum at $z=0$. With this, we can write $(\beta_p \ast P_C)(0)$ in terms of the distance distribution:
    \[
(\beta_p \ast P_C)(0)
        = \sum_{x \in {\Z}_q^n} \beta_p(-x) P_C(x) = \frac{1}{|C|}\sum_{i=0}^n \beta_p(i)\, A_i.
    \]
Substitution into \cref{eq: qsc-inf2} concludes the proof. 
\end{proof}

\section{Eavesdropper's channels with erasures}\label{sec: qEC}

In this section, we consider a wiretap setting in which $\sW_E$ introduces erasures. To set up the analysis, we begin with some notation.

An erasure pattern $Z$ is represented by a pair $(E,S)$, where $E \subseteq [n]$ denotes the set of erased coordinates and $S = X|_{E^\complement}$ records the symbols observed in the nonerased positions. The total number of erasures is denoted as $|E|$ or $|Z|$. 

 An erasure pattern $Z=(E,S)$ is said to be {\bf radially distributed} if, conditioned on $|E| = w$, the erased set $E$ is chosen uniformly among all subsets of size $w$. Only such erasure channels will be considered below.
This class includes several important models. In particular, adversarial erasure patterns with a fixed number of erasures, chosen uniformly over all possible locations, fall into this category.
Moreover, memoryless erasure channels $\qec(p)$ also induce radial erasure
distributions.

We are now ready to state the main theorem of this section.

\begin{theorem}\label{thm: erasure}
Consider a wiretap channel in which $\sW_B$ is noiseless and the $\sW_E$ introduces
radially distributed erasures on an $n$-symbol block.
If $C$ is a UO code, then the coset coding scheme based on
it minimizes each of the following dependence measures over all codes of the same size:
\[
I_2^{\circ}(Z;M), \quad \circ \in \{\mathrm{R},\mathrm{S},\mathrm{A}\}.
\]
\end{theorem}

Unlike the $\qsc$ channels for which all $\alpha$-mutual information measures
coincide, in the present erasure setting, this no longer holds.
As a result, the statement of Theorem~\ref{thm: erasure} does not extend to
all dependence measures considered in the $\qsc$ case, and the minimization
property must be stated explicitly for each admissible notion of mutual
information.

Theorem~\ref{thm: erasure} highlights the power of universally optimal
codes: a single code simultaneously minimizes the dependence between the
message $M$ and the eavesdropper’s observation $Z$ for \emph{all} radial
erasure distributions.
An immediate corollary of this result is obtained by specializing to the
memoryless erasure channel.
\begin{corollary} 
Consider a wiretap channel in which $\sW_B$ is noiseless and $\sW_E=\qec(p)$.  If $C$ is a UO code, then the coset coding scheme 
    based on it minimizes each of the following dependence measures over all codes of the same size:
\[
I_2^{\circ}(Z;M), \quad \circ \in \{\mathrm{R},\mathrm{S},\mathrm{A}\}.
\]
\end{corollary}

As a first step toward proving Theorem~\ref{thm: erasure}, we establish the result for the Rényi mutual information.

\begin{lemma}\label{lem:er1}
Under the assumptions of Theorem~\ref{thm: erasure}, the coset coding scheme
based on a UO code minimizes
$I_2^{\mathrm{R}}(M;Z)$.
\end{lemma}

\begin{proof}

Our goal is to obtain a representation for $I_2^{\mathrm{R}}(M;Z)$ using an energy form relying on a completely monotonic function. 

Let $Z = (E,S)$ with $|E|$ distributed according to some distribution $g$. First, we obtain a simple representation for $P_{Z|M}(z \mid m)$ which is useful in the subsequent calculations. 
\begin{align*}
P_{Z|M}(z \mid m)
    &= P_{E,S|M}(e,s \mid m) \\
    &= P_{S|E,M}(s \mid m,e)\, P_{E|M}(e \mid m) \\
    &= \frac{1}{|C|} \sum_{c \in C} \1\!\left\{ (c + c_m)|_{e^\complement} = s \right\}\frac{g(|e|)}{\binom{n}{|e|}} , 
\end{align*}
where the final line is due to $P_{E|M}(e \mid m) = P_E(e) = \frac{g(|e|)}{\binom{n}{|e|}}$.



To further simplify the representation, we introduce the following notation. For $z = (e,s)$,
\[
    \Gamma(z) := \Gamma(e,s)
    := \sum_{c \in C} \1\!\left\{ c|_{e^\complement} = s \right\},
\]
which counts the number of codewords whose nonerased coordinates agree with the observation $s$.
So we have 
\begin{align}
    P_{Z|M}((e,s) \mid m) = \frac{1}{|C|}\Gamma(e,s-c_m|_{e^\complement})\frac{g(|e|)}{\binom{n}{|e|}},
\end{align}
and as a consequence, we have the following `shift' invariance
\begin{align}
    P_{Z|M}((e,s) \mid m)
    = P_{Z|M}((e, s-c_m|_{e^\complement}) \mid m_0),
\end{align}
where $m_0$ denotes the message corresponding to the identity coset i.e. the code $C$. 

Having this in mind, let us evaluate $I_2^{\mathrm{R}}(M;Z)$.
\begin{align*}
   I_2^{\mathrm{R}}(M;Z) 
   &= \log \sum_{m,z} \frac{P_{MZ}(m,z)^2}{P_M(m)P_Z(z)}\\
  & = \log \sum_m P_M(m) \sum_{z} \frac{P_{Z|M}(z|m)^2}{P_Z(z)}. 
\end{align*}
Observe that for fixed $z = (e,s)$, $$P_Z(e,s) = P_{S|E}(s|e)P_E(e) = \frac{1}{q^{n-|e|}}\frac{g(|e|)}{\binom{n}{|e|}},$$ which only depends on $|e|$. Continuing the calculation,
\begin{align}\label{eq: temp}
    I_2^{\mathrm{R}}&(M;Z) 
    = \log \sum_m P_M(m) \sum_{w=0}^n \sum_{z: |z|=w} \frac{P_{Z|M}(z|m)^2}{P_Z(z)}\nonumber\\
    &= \log \sum_{w=0}^n \frac{q^{n-w}\binom{n}{w}}{g(w)} \sum_m P_M(m) \sum_{z: |z|=w} P_{Z|M}(z|m)^2\nonumber\\
    &= \log \sum_{w=0}^n \frac{q^{n-w}\binom{n}{w}}{g(w)} \sum_{z: |z|=w} P_{Z|M}(z|m_0)^2.
\end{align}

Since for $|z|=w$, $P_{Z|M}(z \mid m_0) = \frac{\Gamma(z)\, g(w)}{|C|\binom{n}{w}}$, we obtain
\begin{align}\label{eq: ssum}
    I_2^{\mathrm{R}}(M;Z) = \log \sum_{w=0}^n \frac{q^{n-w}g(w)}{\binom{n}{w}|C|^2}\sum_{z:|z|=w} \Gamma(z)^2.
\end{align}

A straightforward calculation, e.g., \cite[Lemma 1]{ashikhmin1999binomial}, yields
   $$\sum_{z:|z|=w} \Gamma(z)^2 = |C| \sum_{i=0}^n
      \binom{n - i}{n - w}\, A_i.
      $$

Substituting this back into \eqref{eq: ssum} yields
\begin{align*}
    I_2^{\mathrm{R}}(M;Z) = \log \sum_{w=0}^n \frac{q^{n-w}g(w)}{\binom{n}{w}|C|} \sum_{i=0}^n
      \binom{n - i}{n - w}\, A_i
\end{align*}

Observe that for each $i$, the inner sum corresponds to energy functionals involving
completely monotonic potentials, as characterized in
\eqref{eq: cm_basis}.
Therefore, $I_2^{\mathrm{R}}(M;Z)$ is minimized
by coset coding schemes based on UO codes.
\end{proof}
\vspace*{-.1in}
We now turn to the Arimoto mutual information.
Recall that $I_2^{\mathrm{A}}(Z;M)$ can be written as
\begin{align*}
    I_2^{\mathrm{A}}&(Z;M)
    = H_2(Z) - H_2^{\mathrm{A}}(Z|M) \\
    &= H_2(Z)
       + 2 \log \sum_{m}
       \biggl(
           \sum_{z}
           \bigl(P_M(m)\, P_{Z|M}(z|m)\bigr)^2
       \biggr)^{1/2}.
\end{align*}
Since $P_M$ is uniform and the channel is symmetric across messages, the
inner expression is identical for all $m$, yielding
\begin{align*}
    I_2^{\mathrm{A}}(Z;M)
    = H_2(Z)
      + \log \sum_{z} P_{Z|M}(z|m_0)^2,
\end{align*}
for any fixed message $m_0 \in \sM$.

Using the same arguments as in the Rényi case, the quantity
$\sum_{z} P_{Z|M}(z|m_0)^2$ admits an energy representation in terms of a
completely monotonic potential.
Consequently, UO codes minimize
$I_2^{\mathrm{A}}(Z;M)$ as well.

Finally, we consider the Sibson mutual information.
In this setting, it suffices to show that
\[
I_\alpha^{\mathrm{S}}(Z;M) = I_\alpha^{\mathrm{R}}(Z;M).
\]

To this end, we show that the optimizer $\widetilde{Q}_M^\star$ in the
variational representation of Sibson mutual information coincides with the
uniform distribution over the message set.
We have \eqref{eq:Sibson-opt}
\begin{align}\label{eq:Qm prop}
    \widetilde{Q}_M^\star(m)
    &\propto
    \Biggl(
        \sum_{z}
        \frac{P_M(m)^{\alpha}}{P_Z(z)^{\alpha-1}}
        \, P_{Z|M}(z|m)^{\alpha}
    \Biggr)^{1/\alpha}.
\end{align}

Following steps analogous to those leading to \eqref{eq: temp}, and using 
shift-invariance induced by the coset structure, we simplify the
right-hand side of \eqref{eq:Qm prop} as
\begin{align*}
    \Biggl(
        \sum_{w=0}^{n}
        \Biggl(
            \frac{|C|}{q^{n}}
            \frac{q^{\,n-w}\binom{n}{w}}{g(w)}
        \Biggr)^{\alpha-1}
        \sum_{z:\,|z|=w}
        P_{Z|M}(z|m_0)^{\alpha}
    \Biggr)^{1/\alpha},
\end{align*}
where $m_0$ is an arbitrary but fixed message.
The resulting expression is independent of $m$.
Consequently, $\widetilde{Q}_M^\star$ is uniform over the message set and
coincides with $P_M$.

Therefore, the minimizing distribution in the definition of
$I_\alpha^{\mathrm{S}}(Z;M)$ equals the true marginal, and we conclude that
\[
I_2^{\mathrm{S}}(Z;M) = I_2^{\mathrm{R}}(Z;M).
\]

This completes the proof of \cref{thm: erasure}.

\begin{remark}
    We note that our analysis does not cover the secrecy measures
$I_\alpha^{\mathrm{S}}(M;Z)$ and $I_\alpha^{\mathrm{A}}(M;Z)$.
Although these two quantities coincide in the present setting (since $P_M$
is uniform), they do not appear to coincide with the other dependence
measures considered above.

\end{remark}

\section{Discussion and open problems}

Even though UO codes provide strong optimality guarantees established here, only a small number of code families are currently known to fall in this class. Identifying additional codes with this property remains an important open problem.

Moreover, while the results show that these codes achieve optimal secrecy for the channels and secrecy measures considered, it is unclear whether this condition is also necessary. Determining whether optimal secrecy can be attained by codes other than UO is an interesting and nontrivial question.

Finally, it is natural to ask whether the present approach extends to secrecy measures of other orders. Our analysis relies critically on the appearance of the distance distribution, which enables a characterization in terms of completely monotonic potentials. Whether similar structural reductions are possible for other orders of mutual information remains open. In particular, prior work has used bounds to control Shannon information leakage \cite{LDPCarXiv2025} using 2-\Renyi information leakage; understanding whether such arguments can be strengthened to yield optimality results directly for Shannon mutual information is an important direction for future research.


Another promising direction is the use of universally optimal lattices \cite{Cohn2022} for Gaussian wiretap channels. We expect that these lattices achieve universal secrecy guarantees in the Gaussian setting, mirroring the role played by universally optimal codes in discrete wiretap channels.

{\sc Acknowledgment:} AB was partially supported by NSF grants CCF2330909 and CCF2526035.



\end{document}